\documentclass[journal,twoside,web]{ieeecolor}
\usepackage{generic}

\showboxdepth=\maxdimen
\showboxbreadth=\maxdimen

\usepackage{amsmath, amssymb, amsthm, mathtools, bbm}
\usepackage{cite,cleveref}
\usepackage{graphicx,wrapfig}
\usepackage{algorithm,algpseudocode}
\usepackage[nolist,nohyperlinks]{acronym}
\begin{acronym}
    \acro{CCDF}{complementary cumulative distribution function}
    \acro{HMP}{Hausdorff moment problem}
    \acro{MAB}{multi-armed bandit}
    \acro{PPP}{Poisson point process}
    \acro{SINR}{signal to interference plus noise ratio}
    \acro{TS}{Thompson sampling}
\end{acronym}

\newtheorem{theorem}{Theorem}[]
\newtheorem{proposition}{Proposition}[]

\newtheorem{definition}{Definition}

\crefname{figure}{fig.}{Fig.}
\crefname{section}{sec.}{Sec.}

\newcommand{\bs}[1]{\boldsymbol{#1}}
\newcommand{\cl}[1]{\mathcal{#1}}
\newcommand{\bb}[1]{\mathbb{#1}}

\newcommand{\lb}{\left(}
\newcommand{\rb}{\right)}
\newcommand{\ls }{\left[}
\newcommand{\rs}{\right]}
\newcommand{\lc}{\left\{}
\newcommand{\rc}{\right\}}
\newcommand{\dif}{\mathrm{d}}

\def\BibTeX{{\rm B\kern-.05em{\sc i\kern-.025em b}\kern-.08em
    T\kern-.1667em\lower.7ex\hbox{E}\kern-.125emX}}
\begin{document}
\title{  Channel Access Strategies for Control-Communication Co-Designed Networks}
\author{ Gourab Ghatak, Geethu Joseph, and Chen Quan
\thanks{G. Ghatak is with the Department of Electrical Engineering, IIT Delhi, Delhi 110016 India. Email: gghatak@ee.iitd.ac.in.\\ 
G. Joseph and C. Quan are with the Faculty of Electrical Engineering, Mathematics, and Computer Science, TU Delft, 2628 CD Delft, Netherlands. Emails: \{g.joseph,c.quan\}@tudelft.nl.}
\vspace{-1.5cm}
}

\maketitle

\begin{abstract}
We develop a framework for communication-control co-design in a wireless networked control system with multiple geographically separated controllers and controlled systems, modeled via a Poisson point process. Each controlled system consists of an actuator, plant, and sensor. Controllers receive state estimates from sensors and design control inputs, which are sent to actuators over a shared wireless channel, causing interference. Our co-design includes control strategies at the controller based on sensor measurements and transmission acknowledgments from the actuators for both rested and restless systems - systems with and without state feedback, respectively. In the restless system, controllability depends on consecutive successful transmissions, while in the rested system, it depends on total successful transmissions. We use both classical and block ALOHA protocols for channel access, optimizing access based on sensor data and acknowledgments. A statistical analysis of control performance is followed by a Thompson sampling-based algorithm to optimize the ALOHA parameter, achieving sub-linear regret. We show how the ALOHA parameter influences control performance and transmission success in both system types.
\end{abstract}

\section{Introduction}
The rapid evolution of wireless networks and the rising demand for new applications, such as remote surgery in advanced medical systems, robotics in industrial automation, intelligent buildings, and autonomous driving~\cite{ding2018distributed,lu2023control}, demand real-time control systems over wireless channels. These networked control systems comprise multiple controllers, sensors, actuators, and plants connected through a shared communication medium. Using shared network resources introduces new design challenges that affect the control performance in a networked control system, including unreliable communication links~\cite{hu2020prediction}, transmission delays~\cite{ning2020distributed}, time synchronization issues~\cite{baumann2019control}, and network access constraints and conflicts~\cite{gatsis2015opportunistic}. Furthermore, due to the ad-hoc deployment of wireless access points, especially in the unlicensed bands, the interfering signals from co-channel transmissions of different controllers may severely degrade the performance of a wireless control system. These issues underscore the need for a joint communication-control design of a large-scale wireless networked control system, which we address in this paper.

Prior research has mainly focused on networked control systems with a single controller and plant with one or more actuators and sensors communicating over wireless links. Current literature discusses two principal approaches to designing these systems: independent control-communication design, which treats control and communication components separately~\cite{kawadia2005cautionary,mahmoud2014control,hespanha2007survey}, and control-communication co-design, which integrates both for improved performance~\cite{zhang2014network,qiao2022communication,hu2007stability,cao2012online,girgis2021predictive,wang2020aoi}. This paper focuses on the co-design approach, which is better suited for real-time control applications in wireless networks~\cite{qiao2022communication}. 
Existing studies have explored several aspects of co-design problems, such as stabilizing control systems and enhancing network security with communication imperfections~\cite{zhang2014network} and maintaining system reachability and observability under limited communication resources~\cite{zhang2006communication}. Other works look at handling packet loss in cloud-controlled systems~\cite{qiao2022communication}, missing state information from sensors~\cite{hu2007stability}, and packet loss in sensor-to-controller and controller-to-actuator channels. Another research direction studies the age-of-information metric to assess the reliability and freshness of information in wireless control systems, with\cite{girgis2021predictive} exploring scheduling and power allocation, and \cite{wang2020aoi} optimizing control costs and energy consumption.
In short, co-design methods typically involve two optimization types. The first type optimizes control objectives, such as control costs or control performance metrics like stabilization and tracking error, under communication constraints, such as time delay, bandwidth, packet loss probability, number of communication channels, and sampling period~\cite{chang2019packet,eisen2019control,gatsis2016state,scheuvens2019wireless}, and the second designs communication protocols to meet control performance goals~\cite{park2011wireless,eisen2019control,chang2019optimizing}. 
 While several efforts have been made for joint communication-control designs, large-scale wireless control systems with multiple interfering controllers have not been well-studied in the literature. This paper addresses this issue by considering interfering signals from the randomized locations of control and actuator pairs. 

Our model addresses a wireless networked control system involving multiple controllers and controlled systems located in different geographic areas. The controllers regularly obtain state estimates from sensors, formulate control inputs, and send these inputs to actuators via a shared channel, which results in interference. The co-design approach integrates a control strategy that leverages sensor data and actuator transmission acknowledgments, using a block ALOHA protocol to manage channel access and mitigate interference. It requires an accurate characterization of the interference experienced by the controllers, integrating control and communication theories. Stochastic geometry offers useful tools and techniques to study the statistical impact of different network geometries on the resultant interference. In this work, we use a homogeneous \ac{PPP} to model the locations of different controllers, thereby accounting for the impact of co-channel transmissions. Our main contributions are as follows.
\begin{itemize}
    \item \emph{Communication-control co-design:} We study two types of control systems: restless and rested. In a restless system, control inputs are set to zero if the transmission from the controller fails, while in a rested system, state-based feedback is used if the transmission fails. For both systems, we present a communication-control co-design with a wireless channel acknowledgment-based control input design and a block ALOHA protocol to ensure that the system state is driven to the desired state. Based on the design, we introduce the notion of \emph{block controllability}. For restless systems, block controllability is influenced by the burst length, or the number of consecutive successful transmissions from the controller to the actuator. In contrast, for rested systems, block controllability is determined by the total number of successful transmissions.
    \item \emph{Statistical analysis for the restless system:} For the typical {restless} system , we derive the conditional success probability of a transmission, given a realization of the point process. Additionally, we analyze the conditional distribution $\bar{F}_{{\rm RL,blk}}$ of the burst length. The fraction of controller-controlled system pairs that achieve block controllability is characterized by the meta distribution of the burst length. Since deriving the exact meta distribution is challenging, we characterize the system's performance based on the first moment of $\bar{F}_{{\rm RL,blk}}$ weighted by the channel access probability.
    \item {\emph{Statistical analysis for the rested system:} For the rested controlled system, we derive the conditional success probability of transmission and compare it with the restless controlled system. Leveraging this comparison, we analyze the distribution of successful transmissions for the typical rested system in a given time horizon, based on the meta distribution of success probabilities. Then, we discuss the reconstruction of the approximate meta distribution using a finite number of moments, formulating it as a \ac{HMP}.}
    \item \emph{Learning-based channel access:} To find the optimal online channel access, i.e., to select the optimal ALOHA parameter, we formulate an \ac{MAB} problem, whose reward is based on the number of successful transmissions within a block. The \ac{MAB} is solved using a variant of the \ac{TS} algorithm, which sequentially chooses ALOHA parameters for successive blocks in a centralized manner. We prove that, for both rested and restless systems, regret grows sub-linearly with respect to the number of blocks and the number of slots per block. Our numerical results show how our statistical framework supports the \ac{TS} algorithm and highlight that online learning of the optimal ALOHA parameter can improve network control performance.
    \item {\emph{System design insights:} 
    Our research shows that for a rested system, maintaining a high channel access probability is beneficial, even with increased interference, due to competition for transmission resources. Conversely, a restless system secures access for an entire block after one success, allowing it to reduce access probability to limit interference and boost success. We also discuss how controller density and the required number of successful control inputs impact performance, offering network operators key design and dimensioning guidelines from our co-design framework.}
\end{itemize}
Overall, our research establishes new synergies among networked control systems, stochastic geometry, and learning-based channel access strategies, offering valuable insights from both theoretical and algorithmic perspectives. We make several new contributions compared to the conference version~\cite{ghatak2024poisson}. We extend the analysis of the restless system in \Cref{prop:regretanalysis} to show that the regret of the \ac{TS} algorithm is sub-linear in block length and the number of blocks. Further, we introduce rested control systems, which are more flexible than restless systems, as their controllability depends on total successful transmissions. We analyze the success probability distribution for rested systems using meta distribution. We show that frequent channel access benefits rested systems, while restless systems reduce access frequency to limit interference.

\section{Networked Control System Model}
Consider a large network comprising multiple processes and systems, each with its own independent control mechanisms. Each system in the network is represented by a controller-controlled system pair, where the controlled system consists of an actuator, plant, and sensor.
We model our network using a Poisson bipolar point process consisting of controller-controlled system pairs in the two-dimensional Euclidean plane $\bb{R}^2$~\cite{haenggi2015meta}. The locations of the controllers are modeled as a homogeneous \ac{PPP} $\Phi$ of intensity $\lambda$. The controller-controlled system pairs are indexed as $i=0,1,2,\ldots$, with $r_i$ denoting the distance between the $i$th controller and the typical controlled system. The controller-controlled system pair index $i=0$ refers to the typical pair. Without loss of generality, we assume a typical controlled system at the origin, and the typical controller is at a distance of $r_0$ away from the origin. We note that Slivnyak's theorem~\cite{chiu2013stochastic} guarantees that the process $\Phi$ conditioned on the location of the typical pair is a \ac{PPP} has the same statistics as $\Phi$. 

 The controller periodically receives the system state from the sensor through a dedicated channel. Based on this system state estimate, it computes the control inputs and communicates them to the actuator via a shared wireless link. Since all controller-actuator pairs utilize the same communication resources (time and frequency), interference can occur, resulting in control performance degradation. Our aim is to design a probabilistic channel access protocol that learns and optimizes the control performance of the system. The following subsections elaborate our system model.
\begin{figure}
    \centering
    \includegraphics[width = 0.8\linewidth]{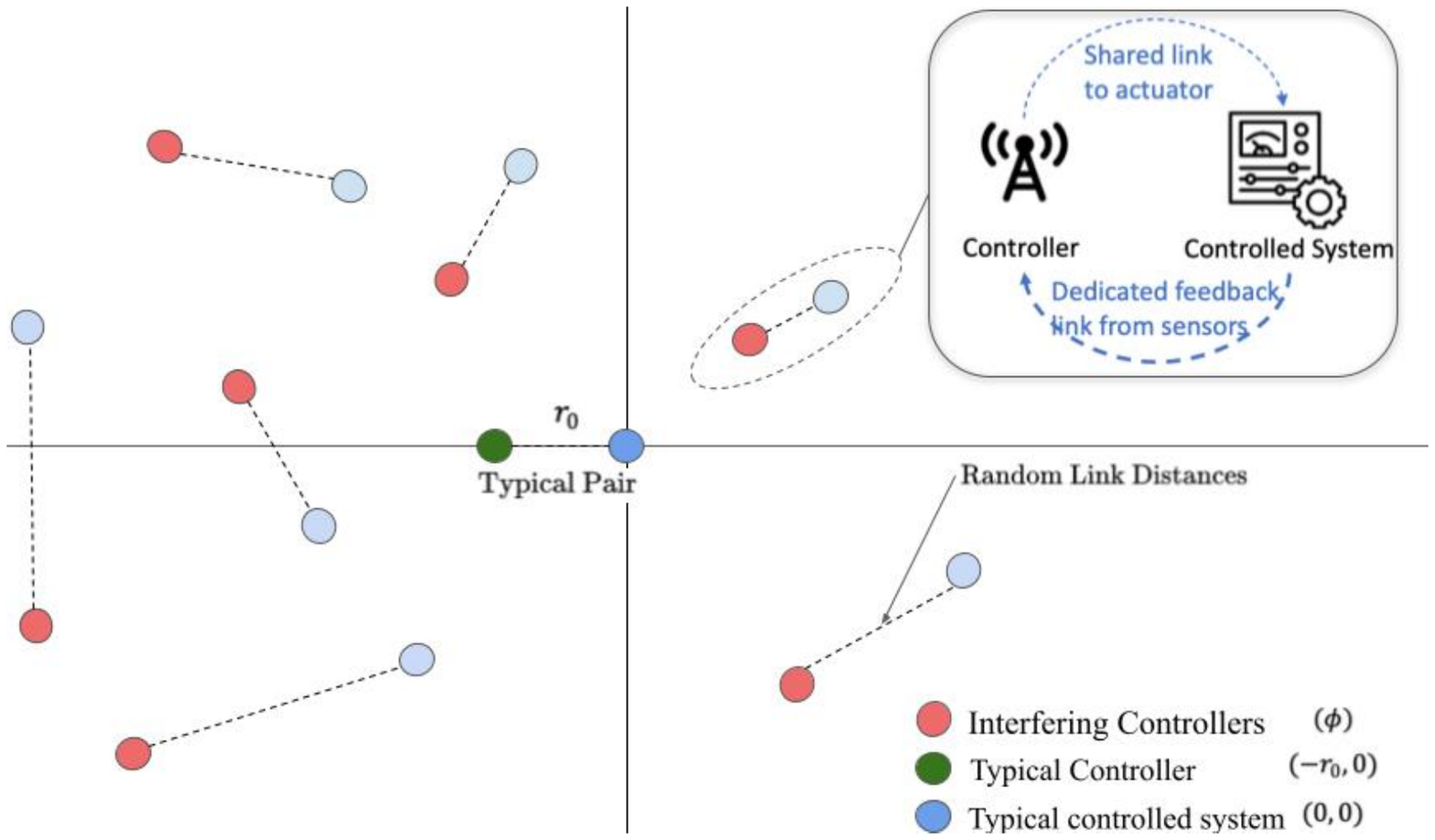}
    \caption{Illustration of the Poisson network of controller-controlled system pairs. Here, the actuators receive the control input from the corresponding controller via a shared link, whereas the sensors send their observations to the controller via a dedicated link.}
    \label{fig:illustration_PPP}
    \vspace{-0.2cm}
\end{figure}
\subsection{Controller-Controlled System Model}
Each controlled system in the network is modeled as a discrete-time linear dynamical system. The typical system is
\begin{equation}\label{eq:controlsys_model}
    \bs{x}(t+1) = \bs{A}\bs{x}(t)+\bs{B}\bs{u}(t)+\bs{v}(t).
\end{equation}
Here, $\bs{x}(t)\in\bb{R}^n$ is the state of the typical controlled system in the network at discrete time $t$ and $\bs{u}(t)\in\bb{R}^m$ and $\bs{v}(t)\in\bb{R}^n$ are its input and process noise, respectively, at time $t\in\bb{Z}_+$. Also, $\bs{A}\in\bb{R}^{n\times n}$ and $\bs{B}\in\bb{R}^{n\times m}$ denote the state and input matrices of the system. Also, we assume that the column space of $\bs{B}$ contains the column space of $\bs{I}-\bs{A}$. The system aims to drive and retain the system state at a desired state $\bs{x}_{\rm des}$ with a suitable choice of control inputs $\bs{u}(t)$.

We assume the actuator computing capabilities are limited,  with all computations being offloaded to the controller. For this, the sensors observe the system state $\bs{x}(t)$ periodically at time $t=kT$ for $k\in\bb{Z}_+$, where the block length $T$ is the delay between two consecutive observations. Here, $k$ is referred to as the block index, and the controller gets the state estimate at the beginning of each block. For simplicity, we assume that the state $\bs{x}(t)$ is communicated to the controller via a reliable link with dedicated resources, ensuring good system state estimate $\hat{\bs{x}}(kT)$ at the controller at time $kT$ for $k\in\bb{Z}_+$. 

The controller estimates the control inputs needed to drive the system to the desired state $\bs{x}_{\rm des}$, which are sequentially communicated to the actuator via a shared link. The controller transmissions are synchronized with the controlled system evolution in \eqref{eq:controlsys_model}. Therefore, the discrete-time index $t$ also denotes the transmission slot index. Due to noise and interference from the other transmitting controllers, the controlled system may not always correctly decode the received signal, leading to an unsuccessful transmission. The success of the transmission depends on the wireless channel, modeled next.

\vspace{-0.2cm}
\subsection{Channel Model and Transmission Success}\label{sec:channelmodel}

Each wireless link experiences fast fading, assumed to be Rayleigh distributed with parameter 1. The fast fading is independent across time and spatially independent across all links. Also,  $\rho$ and $\alpha$ denote the path-loss constant and the path-loss exponent of the channel, respectively~\cite{3GPPChannel}. We denote the channel noise power as $N_0$ and the transmit power as $\eta$. 

A controller's decision to transmit in a time slot $t$ depends on the channel access protocol. Let the channel access state $C_i(t)\in\{0,1\}$ be the indicator variable representing whether the $i$th controller transmits at time $t$, for $i=0,1,\ldots$. Given that the typical controller transmits, the \ac{SINR} at the typical actuator is
\begin{equation}\label{eq:SINR_defn}
    \xi(t) = \frac{\eta\rho |h_0(t)|^2r_0^{-\alpha}}{N_0 + \sum_{i\in\phi} C_i(t)\eta\rho|h_i(t)|^2r_i^{-\alpha}},
\end{equation}
where $h_i(t)$ is the channel fading between the $i$th controller and the typical actuator.

If the \ac{SINR} at time $t$ exceeds a threshold $\gamma>0$ (depending on the application and the receiver hardware), the actuator correctly decodes the received signal, thereby implying a successful transmission. 
The successful transmission by the typical controller at time $t$  is indicated by $S(t)\in\{0,1\}$, i.e., 
\begin{equation}\label{eq:ack_t}
    S(t) = \begin{cases}
        1 &\text{if } C_0(t)\xi(t)>\gamma,\\
        0 &\text{otherwise.}
    \end{cases}
\end{equation}
Also, the controller receives an acknowledgment of transmission $S(t)$ from the actuator. The controller uses $S(t)$ to design the control inputs and its data transmission, as discussed next.

\section{Communication-Control Co-design}
The co-design involves two key components: a control input design based on periodic sensor measurements and transmission acknowledgments, and a random access data transmission scheme that operates independently, without centralized scheduling. We present the co-design approach for two types of controlled systems, one without utilizing direct state feedback control and the other with it, referred to as \emph{restless} and \emph{rested systems}, respectively.

\vspace{-0.2cm}
\subsection{Restless System}
The restless system lacks a direct state-dependent feedback loop, and its actuator applies only the control inputs that are either sent by the controller or are predefined and stored. 
The control input design at the controller is based on its estimate $\hat{\bs{x}}(t)$ of the system state $\bs{x}(t)$. 
Since the noise term $\bs{v}(t)$ is unknown at the controller,  from \eqref{eq:controlsys_model}, this estimate is
\begin{equation}\label{eq:estimate}
    \hat{\bs{x}}(t) = \bs{A}^{t-kT}\hat{\bs{x}}(kT)+\sum_{\tau=kT}^{t-1}\bs{A}^{t-\tau-1} S(\tau)\bs{B}\bs{u}(\tau).
\end{equation}
The controller designs $v$ control inputs $\{\hat{\bs{u}}(t+\tau),\tau=0,1,\ldots,v-1\}$ such that $\bs{x}_{\rm des}=\hat{\bs{x}}(t+v)$, i.e.,
\begin{equation}\label{eq:control_design}
    \bs{x}_{\rm des} = \bs{A}^{v}\hat{\bs{x}}(t)+\sum_{\tau=0}^{v-1}\bs{A}^{v-1-\tau} \bs{B}\hat{\bs{u}}(t+\tau).
\end{equation}
However, we note that for any given states $\bs{x}_{\rm des}$ and $\hat{\bs{x}}(t)$, the above equation has a solution if $v$ exceeds the controllability index of the linear dynamical system in  \eqref{eq:controlsys_model}. Furthermore, the controllability index is further upper-bounded by the degree of the minimum polynomial of $\bs{A}$~\cite{chen1984linear}. So, at time $t=kT$, we choose $v$ as the degree of the minimum polynomial of $\bs{A}$. With this choice, the controller solves \eqref{eq:control_design} using the least squares solution to obtain
\begin{multline}\label{eq:controldesign_est}
    \begin{bmatrix}
        \hat{\bs{u}}(t)^{\top}&
        \hat{\bs{u}}(t+1)^{\top}
        \ldots
        \hat{\bs{u}}(t+v-1)^{\top}
    \end{bmatrix}^{\top} \\= \bs{\Psi}^{\dagger}\ls \bs{x}_{\rm des} - \bs{A}^{v}\hat{\bs{x}}(t)\rs,
\end{multline}
with $t=kT$, and we define
\begin{equation*}
    \bs{\Psi}=\begin{bmatrix}
    \bs{A}^{v-1}\bs{B} & \bs{A}^{v-2}\bs{B} & \ldots & \bs{B}
\end{bmatrix}.
\end{equation*}

 If the transmission of $\hat{\bs{u}}(kT)$ is successful, i.e., $S(kT)=1$, it continues to transmit the next control inputs $\hat{\bs{u}}(kT+1),\hat{\bs{u}}(kT+2),\ldots,\hat{\bs{u}}(kT+L-1)$ until either $S(kT+L-1)=0$ or $L=v$. If $S(kT+L-1)=0$ for some $L\le v$, the controller recomputes the next set of $v$ inputs using \eqref{eq:controldesign_est} with $t=kT+L$ and repeats the above steps. If $S(kT+L-1)=1$ for $L=1,\ldots, v$, all the $v$ control inputs designed by the controller have reached the actuator. Here, $L$ denotes the number of consecutive successful transmissions, referred to as the burst length.
 
 Once all the $v$ control inputs are applied, the state estimate at the controller is $\bs{x}_{\rm des}$. Then, the input required to retain the system state at the desired state $\bs{x}_{\rm des}$ can be computed as
\begin{equation}\label{eq:control_fixed}
        \bar{\bs{u}} = \bs{B}^{\dagger}\lb \bs{I} - \bs{A}\rb\bs{x}_{\rm des}.
\end{equation}
The above control input ensures that $\bs{x}(t+1)=\bs{x}_{\rm des}$ if $\bs{x}(t)=\bs{x}_{\rm des}$ under our assumption that the columns space of $\bs{B}$ contains the column space of $\bs{I}-\bs{A}$. Further, we note that the control input in \eqref{eq:control_fixed} is independent of the system state and can be computed offline. Consequently, the control inputs given by \eqref{eq:control_fixed} can be pre-calculated and stored at the actuator. Once the actuator receives $v$ consecutive control inputs from the controller, it can repeatedly use the inputs from \eqref{eq:control_fixed} without any additional computations. This continues till the end of the block when the sensor sends the state information to the controller at time $(k+1)T$. Overall, in block $k$ the actuator applies the following control inputs,
\begin{equation*}
    \bs{u}(t) = \begin{cases}
        S(t)\hat{\bs{u}}(t) & \text{if }\sum_{\tau=kT}^{t}S(\tau)\le v,\\
        \bar{\bs{u}} & \text{otherwise.}
    \end{cases}
\end{equation*}
Note that $\bs{u}(t)$ is the control input received by the actuator, while $\hat{\bs{u}}(t)$ is the control input sent by the controller. 

Along with the control input design in \eqref{eq:controldesign_est}, the controller also needs to devise a channel access policy. 
Under the above Poisson network model, we aim to optimize channel access policy based on actuator acknowledgment $S(t)$ in \eqref{eq:ack_t} to maximize its probability of successfully steering its state estimate in \eqref{eq:estimate} to the desired state $\bs{x}_{\rm des}$. For the restless system, the controller can drive its state estimate to the desired state $\bs{x}_{\rm des}$ only if $v$ consecutive transmission from the controller to the actuator is successful. This notation of control performance is captured by the random event defined below.
\begin{definition}\label{defn:blockcontrol}
    Consider the typical restless system $(\bs{A},\bs{B})$ with $v$ being the degree of the minimum polynomial of $\bs{A}$. The system is said to \emph{block controllable} in a given block $k$ if there is a run of at least $v$ ones in the sequence $S(kT),S(kT+1),\ldots,S((k+1)T-1)$, where $S(t)\in\{0,1\}$ given by \eqref{eq:ack_t} denotes the success of transmission from the typical controller to the typical actuator at time $t$.
\end{definition}

Having presented the restless system, its control strategy, and controllability, we now describe the rested system before presenting the channel access policy.
\begin{algorithm}
\caption{Control Design and Data Transmission of Typical Controller for the Restless System}
\begin{algorithmic}[1]
    \State \textbf{Parameters:} Block $k$, system matrices $\bs{A}$, $\bs{B}$, desired state $\bs{x}_{\rm des}$, channel access $\{C_0(kT+\tau),\tau=0,1,\ldots,T-1\}$
    \State \textbf{Initialization:} Time $t=kT$, burst length $L=0$ 
    \State Define $v $ as  the degree of the minimum polynomial of $\bs{A}$
    \State Receive $\hat{\bs{x}}(kT)$ from the sensor
        \For {$t=kT,kT+1,\ldots,(k+1)T-1$}
                       
            \If {$C_0(t)=1$}
             \Statex {\it // if the last tx fails, redesign inputs}
                \If {$L=0$}
                    \State Compute $\hat{\bs{u}}(t),\hat{\bs{u}}(t+1),\ldots,\hat{\bs{u}}(t+v-1)$~via~\eqref{eq:controldesign_est}
                                                \Statex {\it // if block controllable, transmit dummy data }
                \ElsIf{$L=v$} 
                    \State Set $\hat{\bs{u}}(t)=\bs{1}$ 
                \EndIf
                \State Transmit $\hat{\bs{u}}(t)$ and receive $S(t)$
                \Statex {\it // if tx fails, reset burst length to 0; otherwise increment}
                \If {$L<v$}
                    \State $L\leftarrow S(t)(L+1)$
                \EndIf
             \EndIf
        \EndFor
\end{algorithmic}
\label{algo:design_restless}
\end{algorithm}

\subsection{Rested System}
The rested system uses a direct state-dependent feedback loop to handle interruptions from missing control signals due to unsuccessful transmission from the controller. The actuators switch between controller inputs and the local feedback loop based on the success of the transmission, as elaborated below.

At the beginning of each block $k$, similar to the restless system, the controller computes $v$ control inputs based on the state estimate $\hat{\bs{x}}(kT)$ from the sensor. These inputs, $\{\hat{\bs{u}}(kT+\tau),\tau=0,1,\ldots,v-1\}$, are given by \eqref{eq:controldesign_est} with $t=kT$ and transmitted sequentially to the actuator.  If the transmission is successful, the corresponding input is applied by the actuator, otherwise, the actuator switches to a state-dependent feedback loop instead of applying zero input.  When $S(t) = 0$, 
\begin{equation*}
    \bs{u}(t) = \bs{B}^{\dagger}\lb \bs{I} - \bs{A}\rb\bs{x}(t),
\end{equation*}
which ensures that the controller's state estimate in \eqref{eq:estimate} does not change, i.e., $\hat{\bs{x}}(t+1)=\hat{\bs{x}}(t)$ when $S(t)=0$. Thus, the control input design at the controller need not be recalculated if the transmission fails. Consequently, the controller keeps attempting to send the same control input until it succeeds. Once all $v$ control inputs from the controller are successfully received by the actuator and applied, the actuator switches to the feedback loop till the end of the block. Overall, in block $k$ the actuator applies the following control inputs,
\begin{equation*}
    \bs{u}(t) = \begin{cases}
        \hat{\bs{u}}(t) & \text{if }S(t)=1 \;\text{,}\;\sum_{\tau=kT}^{t}S(\tau)\!\le v,\\
        \bs{B}^{\dagger}\lb \bs{I} - \bs{A}\rb\bs{x}(t) & \text{otherwise.}
    \end{cases}
\end{equation*}
We reiterate that the key difference between rested and restless systems is that the actuator in the rested system does not use predefined stored inputs if the transmission fails or once all $v$ inputs are received; instead, it switches to feedback-based control. Additionally, the controller computes the control inputs only once per block, with no recomputation after transmission failures. This approach simplifies the control input design process compared to the restless system.  The complete algorithm for data transmission is provided in \Cref{algo:design_restless}.
\begin{algorithm}
\caption{Control Design and Data Transmission of Typical Controller for the Rested System}
\begin{algorithmic}[1]
    \State \textbf{Parameters:} Block $k$, system matrices $\bs{A}$, $\bs{B}$, desired state $\bs{x}_{\rm des}$, channel access $\{C_0(kT+\tau),\tau=0,1,\ldots,T-1\}$
    \State \textbf{Initialization:} Time $t=kT$, Number of successes $\Lambda=0$ 
    \State Define $v$ as  the degree of the minimum polynomial of $\bs{A}$
    \State Receive $\hat{\bs{x}}(kT)$ from the sensor
    \State Compute $\hat{\bs{u}}(kT),\hat{\bs{u}}(kT+1),\ldots,\hat{\bs{u}}(kT+v-1)$~via~\eqref{eq:controldesign_est}
        \For {$t=kT,kT+1,\ldots,(k+1)T-1$}
        \Statex {\it // if not block controllable, transmit the next control input  }
            \If {$\Lambda< v$ and $C_0(t)=1$}
                \State Transmit $\hat{\bs{u}}(kT+\Lambda)$ and receive $S(t)$
                \State Update $\Lambda \leftarrow \Lambda+S(t)$
                                            \Statex {\it // if block controllable, transmit dummy data }
            \ElsIf{$\Lambda= v$ and $C_0(t)=1$} 
                \State Transmit $\hat{\bs{u}}(t)=\bs{1}$ 
            \EndIf
            
        \EndFor
\end{algorithmic}
\label{algo:design_rested}
\end{algorithm}

In the rested system, due to the feedback loop at the actuator, the controller can drive its state estimate to the desired state $x_{\rm des}$ if there are $v$ successful (not necessarily consecutive) transmissions in a block of $T$ slots. Consequently, the notation of control performance is captured by the following event.
\begin{definition}\label{defn:blockcontrolrested}
    Consider the typical rested system $(\bs{A},\bs{B})$ with $v$ being the degree of the minimum polynomial of $\bs{A}$. The system is said to \emph{block controllable} in a given block $k$ if the sequence $S(kT),S(kT+1),\ldots,S((k+1)T-1)$ satisfies $\sum_{\tau=0}^{T-1}S(kT+\tau)\ge v$, where $S(t)\in\{0,1\}$ given by \eqref{eq:ack_t} denotes the success of transmission from the typical controller to the typical actuator at time $t$.
\end{definition}

\subsection{ALOHA-based Channel Access Policy}
After designing the control inputs and transmission scheme, the next crucial step is to develop the channel access state $C_i(t)$, without coordination among the controllers. We note that the controllers lack knowledge of the spatial configuration (density and locations) of interfering controller-system pairs in the network, as well as their transmission states. Therefore, we explore a random channel access strategy, specifically the ALOHA protocol - a widely used multiple-access method for transmitting data over a shared network channel~\cite{baccelli2006aloha}.

We consider two versions of the protocol: classical ALOHA and a modified variant known as \emph{block ALOHA}. In classical ALOHA, each controller accesses the channel in every time slot (not block) with a probability $q\in\cl{P}$ selected from a finite set of $D$ access probabilities $\cl{P} \coloneq \{p_1, p_2, \ldots, p_D\}$. The channel access probability $q$ is called the ALOHA parameter. As a result, the set of controllers transmitting simultaneously may vary from slot to slot within a block. Clearly, this strategy is not suitable for restless systems as it requires consecutive successful transmission, motivating the alternative approach of block ALOHA.  In block ALOHA channel access protocol, each controller is either active or idle during an entire block $k$ with a probability $q\in \cl{P}$. Therefore, the channel access state {$C_i(t)$ for $i=0,1,\ldots$} remains the same for all values of $t=kT,kT+1,\ldots,(k+1)T-1$ within a given block $k$.

Under these protocols, the controller continues to transmit even after $v$ successful transmissions, which are consecutive for the restless system but not necessarily consecutive for the rested system. Following these $v$ transmissions, the controller sends dummy data to the actuator. While the actuators do not use these inputs, they continue to send acknowledgments $S(t)$ upon successfully receiving data from the controller. This feedback allows the controller to learn about its channel conditions and adjust the ALOHA parameter accordingly.
 
 The rest of the paper studies the optimal ALOHA parameter that maximizes the success probability for a typical controller, beginning with a statistical analysis of block controllability.

\section{Statistical Analysis of Controllability}
This section first characterizes the probability of a successful transmission at a given time for a typical pair under the two ALOHA protocols. Building on the above success probabilities, we then characterize the statistics of the controllability metrics: the burst length for the restless system and the total number of successful transmissions for the rested system. 

\subsection{Success Analysis of ALOHA Protocols}
We first look at the block ALOHA protocol assuming that the ALOHA parameter is fixed at $q$ for a given block. We start with the probability of a successful transmission at a given time $t$ for a typical pair. 
\begin{proposition}
    Consider a network following a given PPP $\phi$, whose $i$th controller is at a distance of $r_i$ from the typical actuator, and $C_i(k)$ indicates whether or not it transmits in a given block $k$ of $T$ slots. Given that the typical controller transmits in block $k$, the conditional success probability $P_{\rm blk}$ of the typical controller at a given time $t$ within the block is
    \begin{equation*}
  P_{\rm blk} = e^{-\frac{\gamma N_0}{\eta\rho {r^{-\alpha}_0}}} \prod\limits_{i \in \phi: \;C_i(k)=1} \frac{{r}_0^{-\alpha}}{{r}_0^{-\alpha}+{{\gamma}r_i^{-\alpha}}},
\end{equation*}
where the parameters $P,\rho,N_0$ and $\alpha$ are defined in \eqref{eq:SINR_defn} and $\gamma$ is the SINR threshold for successful transmission in \eqref{eq:ack_t}.
\label{prop:Pblk_defn}
\end{proposition}
\begin{proof}
    See~\Cref{app:_csp}.
\end{proof}

For the block ALOHA case, the set of interfering controllers remains constant throughout the block, so successful transmissions across slots within a block are conditionally independent given the active transmitters. Thus, \Cref{prop:Pblk_defn} accounts for the channel access states $C_i$. For classical ALOHA, given $C_i(t)$, \Cref{prop:Pblk_defn} holds as well. However, in classical ALOHA, successful transmissions across slots within the same block are independent and identically distributed for a given access probability. Therefore,  we average out the randomness in the set of interfering controllers within a block, leading to the following success probability.
\begin{proposition}
Consider a network following a given PPP realization $\phi$, whose $i$th controller is at a distance of $r_i$ from the typical actuator, and $q$ indicates whether or not a controller transmits in a slot of the given block $k$. Given that the typical controller transmits in block $k$, the conditional success probability $P_{\rm cls}$ of the typical controller at a given time $t$ within the block is
    \begin{equation*}
  P_{\rm cls} = e^{-\frac{\gamma N_0}{\eta\rho {r^{-\alpha}_0}}} \prod\limits_{i \in \phi}\ls q\frac{{r}_0^{-\alpha}}{{r}_0^{-\alpha}+{{\gamma}r_i^{-\alpha}}}+1-q\rs,
\end{equation*}
where the parameters $P,\rho,N_0$ and $\alpha$ are defined in \eqref{eq:SINR_defn} and $\gamma$ is the SINR threshold for successful transmission in \eqref{eq:ack_t}.
\end{proposition}
\begin{proof}
From \Cref{prop:Pblk_defn}, the success probability is 
\begin{equation*}
  P_{\rm cls} = \bb{E}\ls e^{-\frac{\gamma N_0}{\eta\rho {r^{-\alpha}_0}}} \prod\limits_{i \in \phi: \;C_i(t)=1} \frac{{r}_0^{-\alpha}}{{r}_0^{-\alpha}+{{\gamma}r_i^{-\alpha}}}\rs.
\end{equation*}
Since each of the $C_i(t)$ is one with probability $q$ and zero with probability $1-q$, the result follows.
\end{proof}

Next, we analyze the restless and rested systems with the two ALOHA protocols, assuming a fixed ALOHA parameter $q$ for each block. 

\subsection{Statistical Analysis of Restless System}
We first look at the restless system with block ALOHA protocol, where conditioned on $\Phi$, the success event in \Cref{prop:Pblk_defn} is independent across the time slots when the typical controller transmits. Also, the probability $P_{\rm blk}$ changes across blocks as $C_i$ changes. We first compute the probability of block controllability, i.e., the burst length $L$ exceeds the desired length~$v$. 
    \begin{proposition}\label{prop:blocksuccess}
    Consider a network of restless systems following a given realization $\phi$ of \ac{PPP} $\Phi$ described in \Cref{prop:Pblk_defn}. Given that the typical controller transmits in block $k$, the probability that the system is block controllable is $\bar{F}_{{\rm RL,blk}}(v)=\bb{P}\lb L \geq v \mid \Phi\rb $ is
    \begin{multline*}
     \bar{F}_{{\rm RL,blk}}(v)  = \!\sum_{l = 1}^{\lfloor \frac{T+1}{v+1} \rfloor} \!(-1)^{l + 1} \!\ls P_{\rm blk} + \frac{T - lv +1}{l} (1 - P_{\rm blk})\rs  \\ \times\binom{T- lv}{l - 1}P_{\rm blk}^{lv} (1 - P_{\rm blk})^{l - 1}. 
    \end{multline*}
\end{proposition}
\begin{proof}
    The result follows from the de Moivre's solution~\cite[Section 22.6]{hald2005history}. 
\end{proof}
The above result establishes the probability of block controllability for a given realization $\phi$ of \ac{PPP} $\Phi$ with a given value of $C_i$ for the controllers. Next, we look at the control performance averaged over the network realization using the moments of conditional success probability.

\begin{theorem}
    Consider a network of restless systems following a PPP $\Phi$ with density $\lambda$ and block ALOHA parameter $q$. The probability of block controllability for the typical pair is
    \begin{multline*}
        P_{\rm RL,blk} = q\bb{E}\ls \bar{F}_{{\rm RL,blk}}(v) \rs = q\sum_{l = 1}^{\lfloor \frac{T+1}{v+1} \rfloor} (-1)^{l + 1} \binom{T - lv}{l - 1}   \\
     \times \bigg( \sum_{{\ell} = 0}^{l - 1} \binom{l-1}{{\ell}} (-1)^{\ell} \zeta(lv + 1 +{\ell}) \\ +  \frac{T - lv +1}{l} \sum_{{\ell} = 0}^l \binom{l}{{\ell}} (-1)^{\ell} \zeta(lv +{\ell}) \bigg).
    \end{multline*}
\label{theo:RL_blk}
Here, with the parameters $P,\rho,N_0$ and $\alpha$ in \eqref{eq:SINR_defn}, $\gamma$ is the SINR threshold for successful transmission in \eqref{eq:ack_t}, we define
\begin{equation*}
            \zeta(l) = e^{-\frac{\gamma N_0l}{\eta\rho r^{-\alpha}_0}} \exp\lb 2\pi\lambda q I(l)\rb,
        \end{equation*}
         with the function $I(l)$ given as
\begin{equation*}
I(l) = \sum_{\ell =1}^{l}\binom{l}{\ell} \int_0^\infty\lb \frac{-\gamma z^{-\alpha}}{r^{-\alpha}_0 + \gamma z^{-\alpha}}\rb^{\ell} \dif z.    
\end{equation*}
\end{theorem}
\begin{proof}
See~\Cref{app:RL_blk}.
\end{proof}

Next, we look at the classical ALOHA protocol, where the ALOHA parameter corresponds to the probability of transmitting in each slot. We obtain a result similar to \Cref{prop:blocksuccess} by replacing $P_{\rm blk}$ with $qP_{\rm cls}$, as follows.
\begin{theorem}
    Consider a network of restless systems following a PPP $\Phi$ with density $\lambda$ and classical ALOHA parameter $q$. The probability of block controllability for the typical pair is
    \begin{multline*}
        P_{\rm RL,cls} = \sum_{l = 1}^{\lfloor \frac{T+1}{v+1} \rfloor} (-1)^{l + 1} \binom{T - lv}{l - 1}   \\
     \times \bigg( \sum_{{\ell} = 0}^{l - 1} \binom{l-1}{{\ell}} (-1)^{\ell} \zeta'(lv + 1 +{\ell}) \\ +  \frac{T - lv +1}{l} \sum_{{\ell} = 0}^l \binom{l}{{\ell}} (-1)^{\ell} \zeta'(lv +{\ell}) \bigg).
    \end{multline*}
\label{theo:RL_cls}
Here, with the parameters $P,\rho,N_0$ and $\alpha$ in \eqref{eq:SINR_defn}, $\gamma$ is the SINR threshold for successful transmission in \eqref{eq:ack_t}, we define
\begin{align*}
            \zeta'(l) &=\bb{E}\ls (qP_{\rm cls})^{l}\rs= q^le^{-\frac{\gamma N_0l}{\eta\rho r^{-\alpha}_0}} \exp\lb 2\pi\lambda q I(l)\rb\\
I'(l) &= \sum_{\ell =1}^{l}\binom{l}{\ell} \int_0^\infty\lb \frac{-q\gamma z^{-\alpha}}{r^{-\alpha}_0 + \gamma z^{-\alpha}}+q-1\rb^{\ell} \dif z.    
\end{align*}
\end{theorem}
We skip the proof as it is similar to that of \Cref{theo:RL_blk}. Further, we can define the meta distribution of the burst length $L$ as the distribution of $\bar{F}_{{\rm RL,blk}}(v))=\bb{P}\lb L \geq v \mid \Phi\rb$ defined in \eqref{prop:blocksuccess} as $\cl{M}_{\rm RL}(v, \beta) = \bb{P}\lb \bar{F}_{{\rm RL,blk}}(v) \geq \beta\rb$.
To clarify, $\cl{M}_{\rm RL}(v, \beta)$ represents the fraction of control systems in the network that experience a burst length of at least $v$ in a sequence of $T$ transmissions in at least $\beta$ fraction of the network realization (or equivalently, due to ergodicity, in at least $\beta$ fraction of transmissions episodes). Unlike the meta distribution of the SINR in wireless networks, we observe that the meta distribution of the burst length is challenging to derive even indirectly via its moments. However, if all the controllers have the same desired burst length $v$, the first moment of the distribution is given by $P_{{\rm RL,blk}}(v)$ in \Cref{theo:RL_blk}. 

\subsection{Statistical Analysis of Rested System}
For the rested system, the controllability metric is $\Lambda$, which is the total number of successes within a block $k$, which is relatively easy to compute. The \acp{CCDF} $\bar{F}_{\rm RD,blk}(v)$ and $\bar{F}_{\rm RD,cls}(v)$ of $\Lambda$ for block and classical ALOHA, respectively, are
\begin{align*}
    \bar{F}_{\rm RD,blk}(v) &= \bb{P}(\Lambda\geq v|\Phi) = q\sum_{l = v}^T \binom{T}{l} (P_{\rm blk})^l (1 - P_{\rm blk})^{T-l}\\
    \bar{F}_{\rm RD,cls}(v) &= \bb{P}(\Lambda\geq v|\Phi) = \sum_{l = v}^T \binom{T}{l} (qP_{\rm cls})^l (1 - q P_{\rm cls})^{T-l}
\end{align*}
Then, we can determine the probability that the typical rested control system is block controllable in at least a specified fraction of network realizations. This probability is denoted as the meta distribution of the number of successes $\Lambda$, as discussed below. We start with the block ALOHA protocol.

\begin{theorem}
    Consider a network of rested systems following a PPP $\Phi$ with density $\lambda$ and block ALOHA parameter $q$. The typical controlled system is block controllable in at least $\beta$ fraction of the network realizations with probability
\begin{equation*}
    \cl{M}_{\rm RD,blk}(v, \beta) = \frac{1}{2} + \frac{1}{\pi} \int_0^\infty \frac{\Im\lb e^{-js \log(P_{\rm blk}(q))} \zeta_{\cl{I}}(s)\rb}{s} {\rm d}s,
\end{equation*}
where $\Im(\cdot)$ represents the imaginary part of the argument and
\begin{align*}
    P_{\rm blk}(q) &= \min\lc p\in[0,1]: \;q\sum_{l = v}^T \binom{T}{l} p^l (1 - p)^{T-l}\geq \beta\rc\\
    \zeta_{\cl{I}}(s) &=   \!e^{-\frac{js \gamma N_0}{\eta\rho r_0^{-\alpha}}}\!\exp \!\ls \! - 2\pi q\lambda\!\!\int_{0}^{\infty}\!\!1\!-\!\lb \frac{{r}_0^   {-\alpha}}{{r}_0^{-\alpha}+{{\gamma}z^{-\alpha}}}\!\rb^{js} \!\!\!{\rm d}z\rs\!.
\end{align*}
\label{theo:RD_blk}
\end{theorem}
\begin{proof}
See~\Cref{app:RD_blk}.
\end{proof}

The next theorem characterizes classical ALOHA.

\begin{theorem}\label{theo:RD_blk2}
    Consider a network of restless systems following a PPP $\Phi$ with density $\lambda$ and classical ALOHA parameter $q$. The probability that the typical pair is block controllable in at least $\beta$ fraction of the network realizations is
\begin{equation*}
    \cl{M}_{\rm RD,cls}(v, \beta) = \frac{1}{2} + \frac{1}{\pi} \int_0^\infty \frac{\Im\lb e^{-js \log(P_{\rm cls}(q))} \zeta_{\cl{I}}'(s)\rb}{s} {\rm d}s,
\end{equation*}
where we define
\begin{align*}
    P_{\rm cls}(q) &= \min\lc p\in[0,1]: \;\!\sum_{l = v}^T \binom{T}{l} (qp)^l (1 - qp)^{T-l}\geq \beta\rc\\
        \zeta_{\cl{I}}'(s) &= e^{-\frac{js \gamma N_0}{\eta\rho r_0^{-\alpha}}}\\
&\times \exp \!\ls - 2\pi\lambda\!\int_{0}^{\infty}\!\!1\!-\!\lb  q\frac{{r}_0^{-\alpha}}{{r}_0^{-\alpha}+{{\gamma}z^{-\alpha}}}+1-q\rb^{js}\!\!{\rm d}z\rs.
\end{align*}
\label{theo:RD_cls}
\end{theorem}

In the rested controlled system, unlike the restless case, the meta distribution of the conditional success probability directly characterizes controllability. However, analytically evaluating the integrals in $\zeta_{\cl{I}}(s)$ and $\zeta'_{\cl{I}}(s)$ from the Gil-Pelaez inversion theorem is generally intractable, and the numerical approximation is often computationally expensive. Recently, the Chebyshev-Markov method was introduced to reconstruct the meta distribution based on a finite sequence of moments~\cite{wang2023fast}. Since the moments of $P_{\rm blk}$ and $P_{\rm cls}$ have already been calculated as an intermediate step in the proof of \Cref{theo:RD_blk} and \Cref{theo:RD_blk2}, respectively, the Chebyshev-Markov method can be formulated as an \ac{HMP}~\cite{mnatsakanov2008hausdorff}. This formulation utilizes the property that if an infinite sequence of moments is monotonic, the distribution of a random variable exists and is unique. 
While exact distribution is based on infinite moments, using finite moments results in an \ac{HMP} reconstruction. For a detailed discussion of this, we refer the reader to~\cite{wang2023fast}.

This concludes our statistical analysis averaged over PPP. However, for a given realization, the optimal ALOHA parameter depends on the locations of the controllers, which is impractical for each controller to know globally. Therefore, the controllers can only learn the optimum value of the ALOHA parameter through the acknowledgments sent by the actuators. In the next section, we discuss a \ac{TS}-based learning algorithm to optimize the ALOHA parameter selection.

\section{\ac{TS}-Based ALOHA Parameter Selection}
In this section, we present a centralized online channel access learning policy based on statistical analysis. In our setting, a central decision-maker broadcasts the block parameter $q(k)\in\cl{P}$ to all the controllers in the network at the beginning of the block $k$. Then, the individual controllers set their channel access states $C_i(t)$ probabilistically based on this parameter. In block ALOHA, $C_i(t)$ remains constant, while in classical ALOHA, $C_i(t)$ varies across the slots. The central decision-maker receives the $T$ acknowledgments for each transmitting (active) controller and updates the channel access probability centrally. The goal of the decision-maker is to sequentially select channel access probabilities $q(1), q(2), \ldots, q(K)$ to maximize the probability of block controllability of the typical pair.

\subsection{Multiarm Bandit Formulation for the Restless System}
We formulate the decision-maker's task of sequentially selecting channel access probabilities across blocks as a \ac{MAB} problem. Here, the set of $D$ access probabilities $\cl{P} \coloneq \{p_1, p_2, \ldots, p_D\}$ represents the arms or actions. We observe the corresponding block controllability for every choice of arm, and the probability of block controllability can be learned over different blocks. 

We next define the reward of \ac{MAB} to reflect block controllability. A naive approach is to assign a reward of one if the typical pair is block controllable, and zero otherwise. However, in block ALOHA, controllability depends on consecutive successful transmissions, while in classical ALOHA, it depends on the total number of successful transmissions. Since the probability of consecutive successes increases with the total number of successes, and the rewards are independent and identically distributed across slots given $\phi$ and $q(k)$, the reward for both protocols can be based on the number of successful transmissions within a block. Consequently, an alternative reward formulation can set the reward to one if the transmission is successful, i.e., $S(t)=1$ and zero otherwise. As the formulation has multiple rewards per block, it leads to faster learning than the naive scheme. 

To complete the \ac{MAB} formulation, we need to define the regret. For a given choice of ALOHA parameters $q^{(K)}=\{q(1),q(2),\ldots,q(K)\}$ and the optimal ALOHA parameter $q_*$, we define the Bayesian regret over $K$ blocks averaged over the realizations of $\Phi$ and the reward sample path $S(t)$ of the algorithm as
\begin{align}
    \cl{R}(q^{(K)}) &= \bb{E}_{\Phi,S(t)} \ls \sum_{k = 0}^{K-1} \sum_{t = kT}^{(k+1)T-1} S(t)\middle| q(k)=q_*\rs \notag\\
&\quad-\bb{E}_{\Phi,S(t)} \ls \sum_{k = 0}^{K-1} \sum_{t = kT}^{(k+1)T-1} S(t)\middle|q(k)=q_k\rs\notag\\
&=\bb{E}_{\Phi} \ls \sum_{k = 0}^{K-1} T( q_*\bar{S}(q_*)-q_k\bar{S}(q_k))\rs, \label{eq:regret_defn}
\end{align}
where $\bar{S}(q)=\bb{E}(S(t)|q(k)=q)$, for $t=kT,kT+1,\ldots,(k+1)T-1$. Naturally, due to the randomness of $\Phi$ (technically, the bandit environment), the optimal ALOHA parameter $q_*$ is a random variable depending on $\Phi$. We next study a Bayesian approach for the \ac{MAB} problem, based on \ac{TS}.

\subsection{Block \ac{TS} for the Restless System}
In the \ac{TS} framework, the decision-maker chooses the ALOHA parameter based on its belief for the expected reward with each access probability in $\cl{P}$. Specifically, the decision-maker starts with a prior probability  $\theta_d$ of obtaining a reward of one when the ALOHA parameter $p_d\in\cl{P}$ is chosen, and a reward of zero with probability $1-\theta_d$~\cite{russo2018tutorial}. Each $\theta_k$ is an action’s success probability or mean reward. In the Bayesian model, $\theta_d$ is modeled as a Beta distribution (conjugate prior for the Bernoulli rewards) with parameters $a_d$ and $b_d$. 
For a selected ALOHA parameter $q(k)=p_d$ in a block $k$, if the typical controller is active, it experiences a reward of $S(t)$ for a given $t$ in the block $k$. Based on this reward, the posterior distribution of the chosen access probability $p_d$ is updated using the Bayes'rule~\cite{russo2018tutorial}. The Bernoulli rewards facilitate a simple update rule: the parameter $a_d$ is incremented by 1 if $S(t)=1$, while the parameter $b_d$ is decreased by one if $S(t)=0$. Nonetheless, the ALOHA parameter does not change until the next block begins. Therefore, we can do a batch update of the parameters at the end of a block. Particularly, at the end of block $k$, the parameter $a_d$ is incremented by $\sum_{t={kT}}^{(k+1)T-1}S(t)$, while the parameter $b_d$ is incremented by $T-\sum_{t=kT}^{(k+1)T-1}S(t)$. Since the typical pair is randomly selected from the distribution of the pairs in the network, the central transmitter can either randomize the pair's selection for observation or consider the average successes across all the pairs in the network. The overall algorithm is summarized in \Cref{algo:main}.

\begin{algorithm}
\caption{TS for Block ALOHA Parameter Selection}
\begin{algorithmic}[1]
    \State \textbf{Parameters:} Beta distribution parameters $\{a_d, b_d\}_{d=1}^{D}$
    \State \textbf{Initialization:} $a_d = b_d = 1, \forall  d \in \{1, 2, \ldots D\}$
    \For {$k = 0,1, 2, \ldots, K$}
        \State Sample $\theta_d \sim \cl{B}(a_d, b_d) \forall d \in \{1, 2, \ldots D\}$
        \State Set parameter $q(k) = p_{d^*}$, where  $d^* = \arg\max_d \theta_d$  
    \State \begin{minipage}[t]{0.8\linewidth}
    Observe acknowledgments $S(kT+\tau)$ for $\tau=0,1,\ldots,T-1$
    \end{minipage} 
    \State Update $a_{d^*}\leftarrow a_{d^*} + \sum_{\tau=0}^{T-1}S(kT+\tau)$
    \State Update $b_{d^*} \leftarrow b_{d^*} + T - \sum_{\tau=0}^{T-1}S(kT+\tau)$.
    \EndFor
\end{algorithmic}
\label{algo:main}
\end{algorithm}

The Bayesian regret for the block \ac{TS} algorithm is bounded, as presented next.
\begin{proposition}\label{prop:regretanalysis}
    The Bayesian regret of the block \ac{TS} algorithm in \Cref{algo:main} after $K$ blocks, defined in \eqref{eq:regret_defn},  is bounded as
    \begin{equation*}
        \cl{R}_{\rm TS}(K) \leq \cl{O}\lb \sqrt{TKD\log(K)}\rb,
    \end{equation*}
    where $T$ is the number of time slots per block and $D$ is the number of choices for access probabilities.
    \label{prop:regret}
\end{proposition}
\begin{proof}
    See~\Cref{app:regret}.
\end{proof}
The result aligns with the classical TS, highlighting the dependence of regret on the system parameters $K$, $T$, and $D$. However, it does not capture the dependence on the \ac{PPP} density $\lambda$. Deriving the statistics of the frequentist regret requires addressing the challenges posed by the randomness of $\Phi$, which is deferred to future work.

\section{Numerical Results and Discussion}
\label{sec:NRD}
In this section, we present some numerical results to highlight the salient features of our analysis. Unless otherwise stated, we assume a transmit power of $\eta = 24$~dBm, a path-loss exponent of $\alpha = 2$ considering line-of-sight propagation. The carrier frequency is assumed to be 3.2~GHz with an operating bandwidth of 200~MHz. The distance between the typical controller-controlled system pair is set as 10~m.

\begin{figure}
    \centering
    \includegraphics[width=0.6\linewidth]{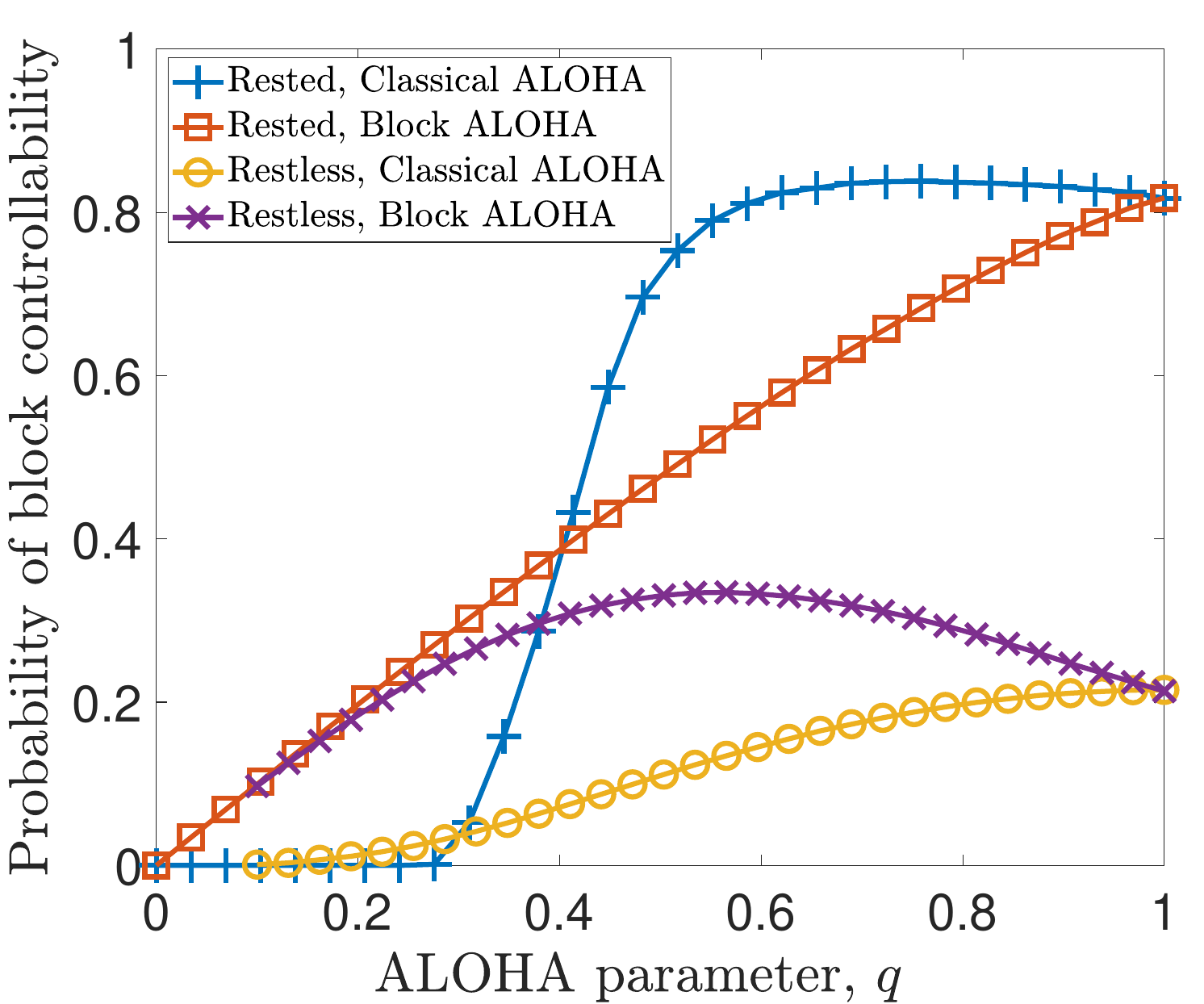}
    \caption{Probability of block controllability for the rested and restless systems with block or classical ALOHA channel access schemes. Here $T = 20$, $v = 4$, and $\lambda = 5e-3$ m$^{-2}$.}
    \label{fig:Classic_Block_Rested_Restless}
    \vspace{-0.3cm}
\end{figure}
We empirically study the effect of classical and block ALOHA parameters on the probability of block controllability for the rested and restless systems. \Cref{fig:Classic_Block_Rested_Restless} shows the probability of block controllability for the rested and restless systems for a single block ($K = 1$) averaged across 10000 realizations of \ac{PPP}. With classical ALOHA, both the rested and restless systems have a lower probability of controllability when the ALOHA parameter $q$ is low ($q \leq 0.3$),  as a lower channel access frequency limits controllability performance. As $q$ increases, controllability improves but declines at high $q$ due to interference, reducing successful transmission probability. The block ALOHA also has a similar trend. However, the increase in the probability of block controllability here is sharper as compared to classical ALOHA, as a single successful access guarantees channel access across all block slots. We note that although for the selected values of the system parameters, the controllability of the rested system with block ALOHA channel access shows a monotonic increase, for a denser deployment of control systems (higher $\lambda$), a high $q$ may lead to a decrease in the probability of block controllability. At $q = 1$, both protocols yield similar performance.

Comparing the two ALOHA protocols for the restless system, block ALOHA outperforms classical ALOHA across all $q$ values. This observation is intuitive, as block ALOHA enables consecutive channel access, which increases the likelihood of successive successful transmissions, thereby leading to a higher probability of achieving block controllability in the restless system. Conversely, for the rested system, the choice between block and classical ALOHA is non-trivial. For example, in \Cref{fig:Classic_Block_Rested_Restless}, we see that $q \approx 0.65$ under classical ALOHA achieves higher controllability than block ALOHA’s optimal $q=1$, an effect attributed to Jensen's inequality. Specifically, classical ALOHA  adds randomness through varying transmitting controllers for each slot, whereas block ALOHA fixes the set of transmitting controllers for the entire block. So, for different values of $q$, the expectation of the block controllability over the transmitting set may be greater or lower than the block controllability for an expected transmission set. Moreover, rested systems generally exhibit a higher probability of block controllability at any given $q$ value when using the same ALOHA protocol, as they are more flexible and incorporate an additional feedback loop.

\begin{figure}
    \centering
    \includegraphics[width = 0.5\linewidth]{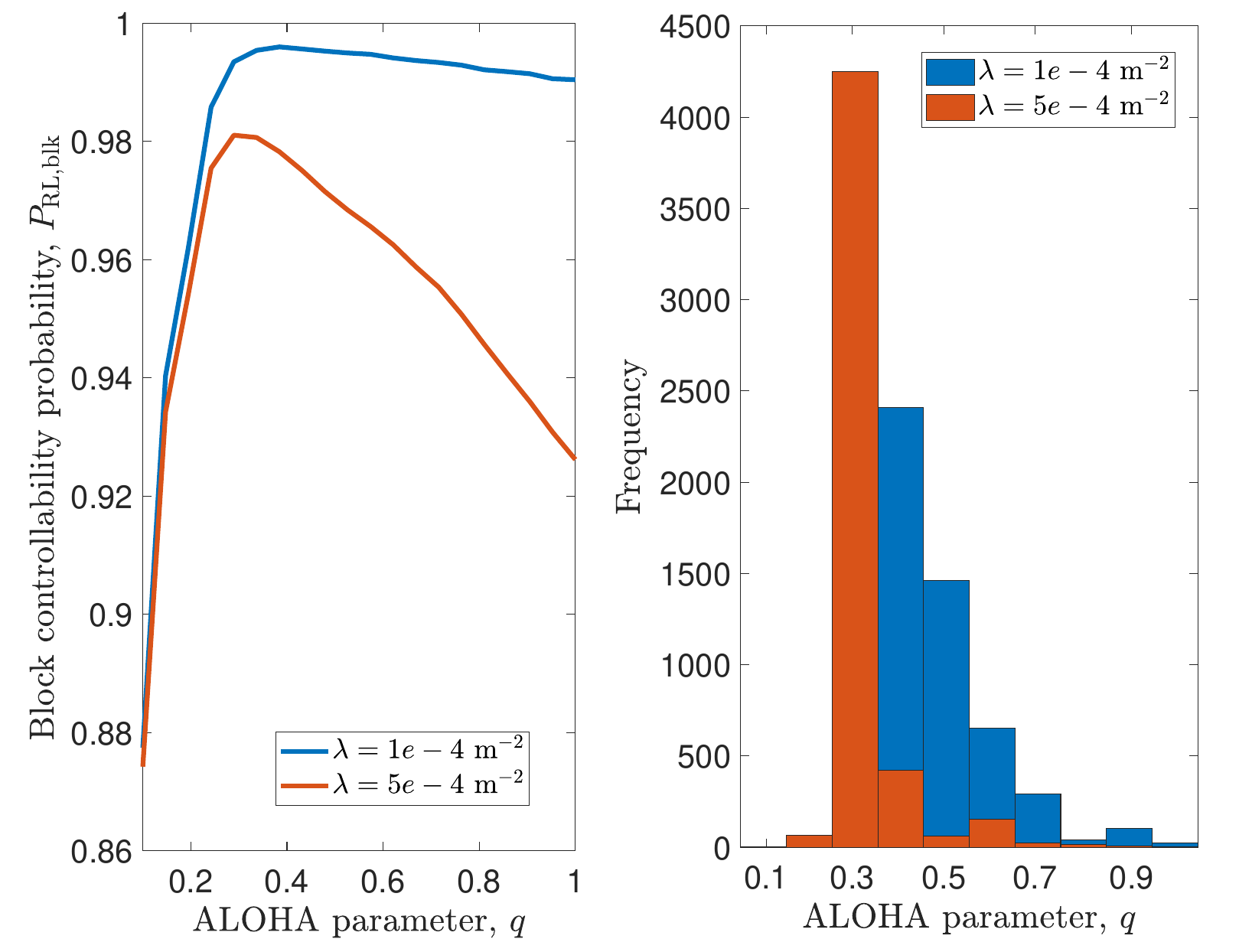}
    \caption{Comparison of block controllability probability of restless system with $v = 5$ in a block length of $T = 20$ and the relative selection of the ALOHA parameters by the block \ac{TS} framework in $K = 5000$ blocks.}
    \label{fig:TS1}
    \vspace{-0.2cm}
\end{figure}
We next study the impact of the block and classical ALOHA parameters on block controllability of the restless and rested systems for a large value of $K=5000$. Here, unlike the empirical study in \Cref{fig:Classic_Block_Rested_Restless}, we use the theoretical expressions derived in \Cref{theo:RL_blk,theo:RL_cls,theo:RD_blk,theo:RD_cls}. Since the analysis of block and classical ALOHA are similar, we focus on illustrating the impact of block ALOHA parameter on the block controllability in the restless system and the impact of classical ALOHA parameter on block controllability in the rested system. These results are summarized in \Cref{fig:TS1,fig:TS2}. 

\Cref{fig:TS1} shows the theoretical expression for block controllability probability $P_{\rm RL,blk}$ of restless system in \Cref{theo:RL_blk} for two different network densities with $\lambda = 10^{-4}$ m$^{-2}$ and $\lambda = 5\times10^{-4}$ m$^{-2}$. The trends are consistent with those in \Cref{fig:Classic_Block_Rested_Restless}, where low channel access probability $q$ results in low $P_{\rm RL,blk}$ despite low interference.  Further, $P_{\rm RL,blk}$ increases with $q$ until a threshold is reached, beyond which higher $q$ leads to increased interference, degrading success probability despite more frequent channel access. \Cref{fig:TS1} also shows that a higher density of the network, indicated by higher $\lambda$, results in lower $P_{\rm RL, blk}$ due to an increase in interference. Also, $P_{\rm RL, blk}$ decreases at a faster rate with an increase in $q$, as compared to a lower value of $\lambda$, as the intensity of interfering controllers is $q\lambda$. Furthermore, the optimal channel access probability depends on the intensity of the control systems; for example, a lower value of $\lambda= 10^{-4}$~m$^{-2}$ leads to a higher optimal channel access probability $\approx 0.35$ as compared to the optimal value of $\approx 0.25$ when $\lambda = 5\times10^{-4}$~m$^{-2}$. 

The right panel of \Cref{fig:TS1} shows the results from \ac{TS} with  $D=10$ and $\cl{P}=\{0.1, 0.2, \ldots, 1\}$. From the plot (on the left panel of \Cref{fig:TS1}), note that the optimal ALOHA parameters from $\cl{P}$ are $0.4$ and $0.3$ for $\lambda = 10^{-4}$ m$^{-2}$ and $5\times10^{-4}$ m$^{-2}$, respectively. We observe that the \ac{TS} algorithm efficiently learns this parameter relatively quickly. In particular, in $K=5000$ blocks, the system selects the optimum ALOHA parameter in more than 4000 blocks for $\lambda = 5\times10^{-4}$ m$^{-2}$ where the mean rewards for the different arms are significantly distinct. On the contrary, for $\lambda = 10^{-4}$ m$^{-2}$, several non-optimum ALOHA parameters have mean rewards close to the optimum value and hence, the optimum ALOHA parameter is chosen less frequently.

\begin{figure}
    \centering
    \includegraphics[width = 0.5\linewidth]{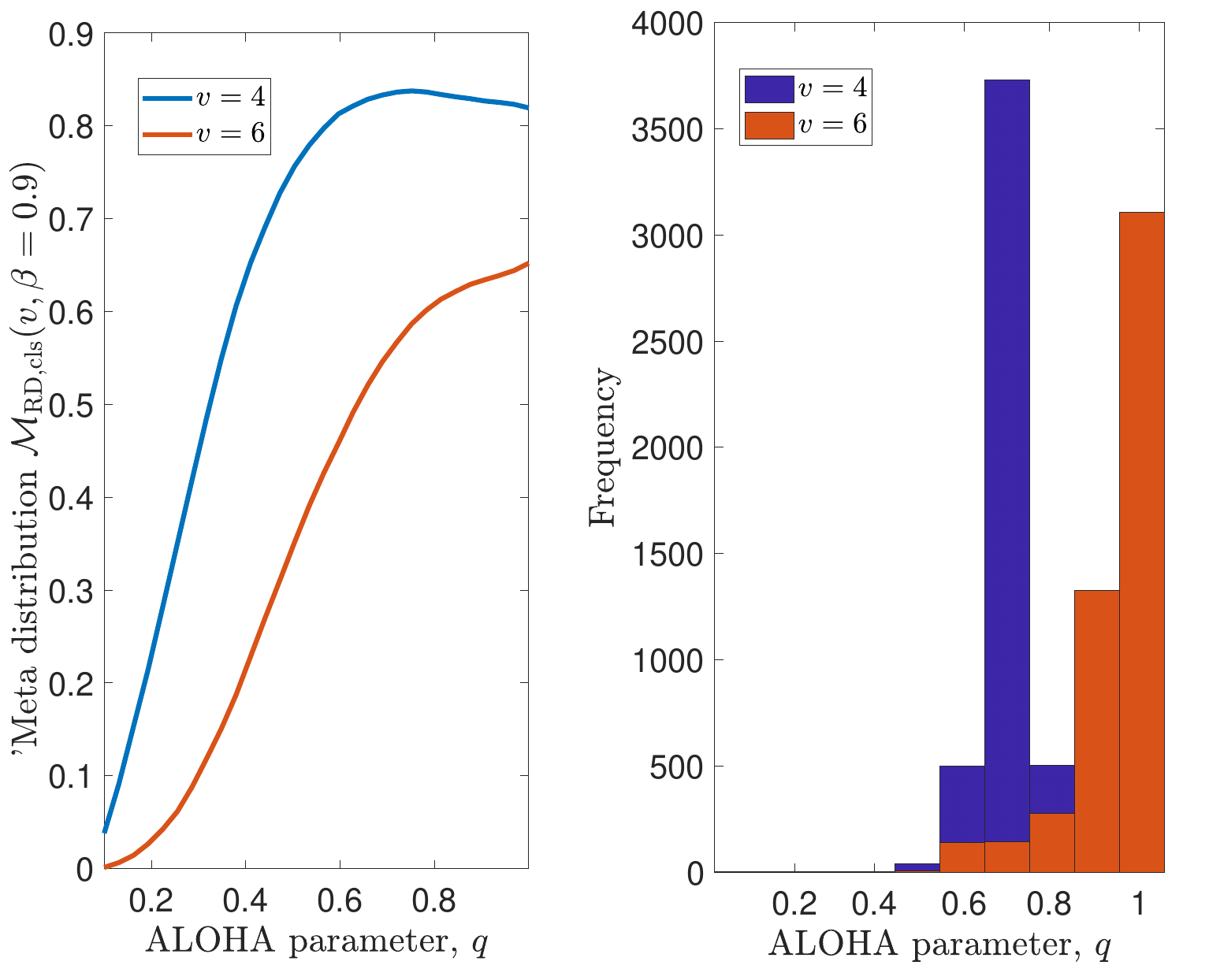}
    \caption{Comparison of the meta distribution $\cl{M}_{\rm RD,cls}(v, \beta=0.9)$ of rested system with classical ALOHA and the relative selection of the classical ALOHA parameter by the \ac{TS} framework in $K=5000$ slots.}
    \vspace{-0.5cm}
    \label{fig:TS2}
\end{figure}

\Cref{fig:TS2} shows the meta distribution $\cl{M}_{\rm RD,cls}(v, \beta)$  of rested system with classical ALOHA in \Cref{theo:RD_cls} for two different controllability indices $v=4$ and 6 with reliability threshold of $\beta=0.9$. The higher value of $v$ makes the system more demanding, resulting in a lower meta distribution. For $v=6$, the optimal ALOHA parameter is 1, prioritizing channel access over increased interference. The optimal parameter decreases as $q$ decreases, indicating that the rested system should reduce its channel access probability for lower $v$ to achieve the required successful transmissions while minimizing interference. 
The right panel of \Cref{fig:TS2} shows the performance of the classical \ac{TS} algorithm in determining the optimal classical ALOHA parameter with respect to different $v$. Here, we consider blocks of $20$ slots and observe the system over $250$ blocks. Similar to the restless case, we consider $\cl{P}=\{0.1,0.2,\ldots,1\}$. For $v = 4$, our analytical framework predicts that the optimal access probability is $q = 0.7$, while for $v = 6$, the optimum scheme is to transmit in all slots. This is reflected in the frequency with which the \ac{TS} framework selects the channel access probability. However, the performance of the rested system using classical \ac{TS} is inferior to that of block \ac{TS} for the restless system, indicating that block \ac{TS} consistently outperforms classical \ac{TS} for all systems and both protocols.
\begin{figure}
    \centering
    \includegraphics[width = 0.6\linewidth]{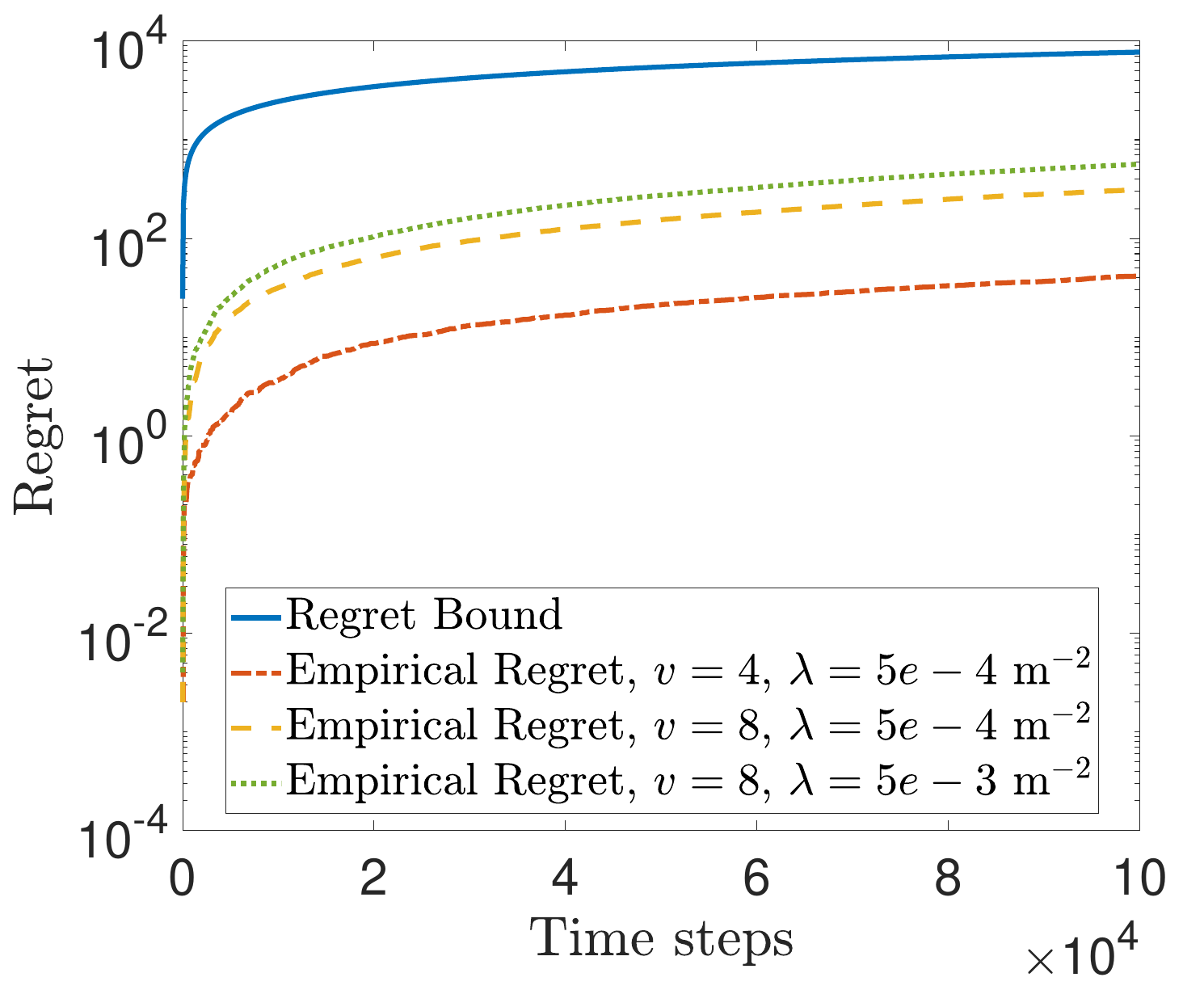}
    \caption{Regret for the restless system for different values of $v$ and $\lambda$.}
    \label{fig:TS_Regret_}
    \vspace{-0.2cm}
\end{figure}

Finally, we illustrate the evolution of regret of \ac{TS} for the restless system, defined in \Cref{fig:TS_Regret_}. Higher $\lambda$ and $v$ lead to a higher regret due to intensified interference and decreased controllability, respectively. Nonetheless, the regret evolves sub-linearly in both cases, aligning with the bound in \Cref{prop:regret}. 

\section{Conclusion}
We developed a control design and communication protocol for two types of Poisson networked control systems: restless and rested. Using stochastic geometry, we analyzed interference impacts on performance, showing that dense networks with simultaneous channel access degrade performance. To mitigate this, we introduced block and classical ALOHA protocols for channel access. Our study explored how ALOHA parameters affect controllability probability, identifying optimal parameters based on system characteristics, such as controller density. We also developed a \ac{TS} algorithm for sequentially optimizing these parameters, demonstrating sub-linear regret. Future work will focus on characterizing the regret of \ac{TS} across various network realizations and analyzing decentralized ALOHA parameter selection.

\appendices
\crefalias{section}{appendix}

\section{Proof of \Cref{prop:Pblk_defn}}
\label{app:_csp}
Given the distances $r_i$'s and using \eqref{eq:SINR_defn}, the success probability $P_{\rm blk}=\bb{P}(\xi(t)>\gamma \mid \Phi  )$ is given by
\begin{align*}
  P_{\rm blk} &  = \bb{P}\lb \!|h_0(t)|^2 \!> \!\frac{\gamma[N_0+\sum_{i \in \phi} \eta\rho|{h}_i(t)|^2r_i^{-\alpha}C_i(k)]}{ P \rho r_0^{-\alpha}} \!\rb  \\
    & \overset{(a)}{=} e^{-\frac{\gamma N_0}{\eta\rho r_0^{-\alpha}}}\bb{E}\ls \exp\lb \!\frac{{-\gamma}\sum_{i \in \phi}\rho r_i^{-\alpha}}{   r_0^{-\alpha}}|{h}_i(t)|^2C_i(k)\!\rb \rs  \\
    &  \overset{(b)}{=}   e^{-\frac{\gamma N_0}{\eta\rho r_0^{-\alpha}}}  \prod_{i \in\phi} \bb{E}\ls \exp\lb \frac{{-\gamma}r_i^{-\alpha}}{r_0^{-\alpha}}|{h}_i(t)|^2C_i(k)\rb \rs,   
 \end{align*}
 where Step $(a)$ is because $|h_0(t)|^2$ is exponential distributed. Step $(b)$ follows as $h_i(t)$'s are independent across different values of $i$. Finally, the desired result follows from the Laplace transform of the exponentially distributed $|h_i(t)|^2$. $\hfill\blacksquare$
\section{Proof of \Cref{theo:RL_blk}}
\label{app:RL_blk}
    From the definition of $\bar{F}_{{\rm RL,blk}}(v) $ in \Cref{prop:blocksuccess}, its expected value of  is given by
\begin{multline*}
    \bb{E}\ls \bar{F}_{{\rm RL,blk}}(v) \rs = \sum_{l = 1}^{\lfloor \frac{T+1}{v+1} \rfloor} (-1)^{l + 1} \binom{T - lv}{l- 1} \\
     \times\bb{E}\ls P_{\rm blk}^{lv+1} (1 - P_{\rm blk})^{l-1} 
     +\frac{T - lv +1}{l} P_{\rm blk}^{lv} (1 - P_{\rm blk})^{l}\rs.
\end{multline*}
By the Binomial expansion $(1 - P_{\rm blk})^{l} = \sum_{\ell = 0}^{l} \binom{l}{\ell} (-1)^{\ell} P_{\rm blk}^{\ell}$, we arrive at
\begin{multline}
    \bb{E}\ls \bar{F}_{{\rm RL,blk}}(v) \rs = \sum_{l = 1}^{\lfloor \frac{T+1}{v+1} \rfloor} (-1)^{l + 1} \binom{T - lv}{l - 1}   \\
     \times \Bigg(\sum_{{\ell} = 0}^{l - 1} \binom{l-1}{{\ell}} (-1)^{\ell} \bb{E}\ls P_{\rm blk}^{lv + 1 +{\ell}}\rs \\ +  \frac{T - lv +1}{l} \sum_{{\ell} = 0}^l \binom{l}{{\ell}} (-1)^{\ell} \bb{E}\ls P_{\rm blk}^{lv +{\ell}}\rs\Bigg).\label{eq:expected_F} 
\end{multline}
We compute the moments of conditional success 
probability~as 
\begin{equation*}
    \bb{E}\ls P_{\rm blk}^l\rs = e^{-\frac{l\gamma N_0}{\eta\rho r_0^{-\alpha}}}\bb{E}\ls \prod\limits_{i \in\phi: C_i=1}   \lb \frac{{r}_0^   {-\alpha}}{{r}_0^{-\alpha}+{{\gamma}r_i^{-\alpha}}}\rb^l \rs.
\end{equation*}
Here, transmitting controllers follow a \ac{PPP} of intensity $\lambda q$, and using the probability generating functional of \ac{PPP}~\cite{haenggi2012stochastic},
\begin{align*}
    & \bb{E}\ls \lb \prod\limits_{i \in\phi: C_i=1}   \frac{{r}_0^   {-\alpha}}{{r}_0^{-\alpha}+{{\gamma}r_i^{-\alpha}}}\rb^l \rs \\
      & =\exp\lb - 2\pi\lambda q \int_0^\infty 1 - \lb \frac{{r}_0^{-\alpha}}{{r}_0^{-\alpha}+{{\gamma}z^{-\alpha}}}\rb^l \dif z\rb  \\
     &=\exp\lb - 2\pi\lambda  q\int_0^\infty 1 - \lb 1 - \frac{{\gamma}z^{-\alpha}}{r_0^{-\alpha} + {\gamma}z^{-\alpha}}\rb^l \dif z\rb.
\end{align*}
Further, from the binomial expansion 
\begin{equation}
    (1 - y)^{l} -1= \sum_{\ell = 1}^{l} \binom{l}{\ell} (-y)^{\ell},
\end{equation}
with $y=\frac{{\gamma}z^{-\alpha}}{r_0^{-\alpha} + {\gamma}z^{-\alpha}}$, we get $\zeta(l)=\bb{E}\ls P_{\rm blk}^l\rs$. The desired result is obtained by substituting this relation in \eqref{eq:expected_F}. $\hfill\blacksquare$

\section{Proof of~\Cref{theo:RD_blk}}
\label{app:RD_blk}
From the definition of $\cl{M}_{\rm RD,blk}(v, \beta)$, we have
\begin{align*}
    \cl{M}_{\rm RD,blk}(v, \beta) &=\bb{P}\lb q\sum_{l = v}^T \binom{T}{l} P_{\rm blk}^l(1 - P_{\rm blk})^{T - l} \geq \beta\rb \\
    &= \bb{P}\lb P_{\rm blk} \geq P_{\rm blk}(q)\rb,
\end{align*}
which is the \ac{CCDF} of random variable $P_{\rm blk}$. Now, we use  the Gil-Pelaez theorem~\cite{gil1951note}, we derive
\begin{equation*}
    \cl{M}_{\rm RD,blk}(v, \beta) = \frac{1}{2} + \frac{1}{\pi} \int_0^\infty \frac{\Im\lb e^{-js \log(\beta)} \bb{E}\ls P_{\rm blk}^{js}\rs\rb}{s} {\rm d}s.
\end{equation*} 
We complete the proof by evaluating the moments $\zeta_{\cl{I}}(s)=\bb{E}\ls P_{\rm blk}^{js}\rs$ using \Cref{prop:Pblk_defn} and the transmitting controllers form a PPP with density $q\lambda$ as
\begin{align*}
    \zeta_{\cl{I}}(s)&= e^{-\frac{js \gamma N_0}{\eta\rho r_0^{-\alpha}}}\bb{E}\ls \prod\limits_{i \in\phi}   \lb \frac{{r}_0^   {-\alpha}}{{r}_0^{-\alpha}+{{\gamma}r_i^{-\alpha}}}\rb^{js} \rs\\
    &= e^{-\frac{js \gamma N_0}{\eta\rho r_0^{-\alpha}}}\!\exp \ls \! - 2\pi q\lambda\!\!\int_{0}^{\infty}\!\!1-\!\lb \frac{{r}_0^   {-\alpha}}{{r}_0^{-\alpha}+{{\gamma}z^{-\alpha}}}\!\rb^{js} \!{\rm d}z\!\rs.
\end{align*}

\section{Proof of \Cref{prop:regret}}
\label{app:regret}
The proof is based on the analysis of Thompson sampling for finite-armed Bandits~\cite[Chapter 36]{lattimore2020bandit}. We define $\mu_*=q_*T\bar{S}(q_*)$ as the average reward in a block with the optimal channel access probability $q_*$, and $\mu(p_d)=p_dT\bar{S}(p_d)$ as the average reward obtained when the channel access probability is $p_d$, for $d=1,2,\ldots,D$. Then, from \eqref{eq:regret_defn}, the reward is
\begin{equation*}
    \cl{R}(K) = \bb{E}_{\Phi} \ls \sum_{k=0}^{K-1} \bb{E} \lb \mu_* - \mu(q(k))\rb\rs,
\end{equation*}
where $q(k)$ is the channel access probability in block $k$. Now, by the law of total expectation, for any random event $\cal{E}$, 
\begin{align}
    \cl{R}(K) &= \bb{E}_{\Phi} \ls \bb{P}(\cl{E})\bb{E}\ls\sum_{k=0}^{K-1} \bb{E} \lb \mu_* - \mu(q(k)) |\cl{E}\rb\rs\rs\notag\\
    &\hspace{0.5cm}+ \bb{E}_{\Phi} \ls \bb{P}(\cl{E}^{\complement})\bb{E}\ls\sum_{k=0}^{K-1} \bb{E} \lb \mu_* - \mu(q(k))|\cl{E}^{\complement}\rb\rs\rs\notag\\
    &\leq \cl{R}_1(K) + KT\bb{P}(\cl{E}^{\complement}),\label{eq:regret_bound}
\end{align}
where $\cl{R}_1(K)=\bb{E}_{\Phi} \ls\sum_{k=0}^{K-1} \bb{E} \lb \mu_* - \mu(q(k)) |\cl{E}\rb\rs$ and we use the property that $\mu_* - \mu(q(k))\leq \mu_* \leq T$.

We next define the event $\cl{E}$ using the empirical estimate of the reward computed by the \ac{TS} algorithm. Let the empirical estimate of the reward corresponding to the access probability $p_d$ after the $k$th block be
\begin{equation*}
    \hat{\mu}_d(k)= \frac{1}{\Gamma(k,d)}\sum_{\kappa=0}^{k-1}\mathbbm{1}(q(k)=p_d)p_d\sum_{t=\kappa T}^{(\kappa+1)T-1}S(t),
\end{equation*}
where $\Gamma(k,d)$ denotes the number of blocks until block $k$ in which channel access probability is chosen as $p_d$.  If $\Gamma(k,d)=0$, we define $\hat{\mu}_d(k)=0$.
Here, given the channel access probability sequence $q^{(K)}=\{q(1),q(2),\ldots,q(K)\}$, the Bernoulli random variable $S(t)$ has mean $P_{\rm blk}(p_d)$ defined in \Cref{prop:Pblk_defn} with $q=p_d$. Hence, $\bb{E}(\hat{\mu}_d(k)|q^{(K)})=\mu(p_d)$. We define $\cl{E}=\cap_{k=0}^{K-1}\cap_{d=1}^D\cl{E}(k,d)$, where the event $\cl{E}(k,d)$ is that for a given $k=0,1,\ldots,K-1$ and $d=1,2\ldots,D$,
\begin{equation*}
    |\hat{\mu}_d(k-1) - \mu(p_d)| < \epsilon(k,d) = \sqrt{\frac{2 \log (1/\delta)}{\max
    \{1, T\Gamma(k,d)\}}},
\end{equation*}
for some $0<\delta<1$. Now, to compute the bound in \eqref{eq:regret_bound}, we first compute $\bb{P}(\cl{E}^{\complement})$ using the union bound as follows:
\begin{equation*}
    \bb{P}(\cl{E}^{\complement}) \leq \sum_{k=0}^{K-1}\sum_{d=1}^{D}\bb{P}\lb\cl{E}(k,d)^{\complement}\rb
    \leq 2KD\delta,
\end{equation*}
where we also use Hoeffding's inequality. So, \eqref{eq:regret_bound} implies
\begin{equation}
    \cl{R}(K) \leq \cl{R}_1(K) + 4K^2TD\delta.\label{eq:regret_bound_new}
\end{equation}

Now, we bound the term $\cl{R}_1(K)$ in \eqref{eq:regret_bound_new} by defining the $\sigma-$algebra generated by the ALOHA parameters and the corresponding rewards by the end of the $k$th slot as $
    \cl{F}_k = \sigma\lb q^{(k)}, S(0),S(1),\ldots,S((k+1)T-1)\rb$.
Then, we have
\begin{equation*}
    \cl{R}_1(K) =  \bb{E}_{\Phi} \ls \sum_{k=0}^{K-1} \bb{E} \lb \mu_* - \mu(q(k)) | \cl{F}_{k-1},\cl{E}\rb\rs.
\end{equation*}
However, the \ac{TS} algorithm implies that the conditional distribution of $q_*$ is the same as $q(k)$~\cite{lattimore2020bandit}, i.e., $\bb{E} \ls \nu(\mu_*)| \cl{F}_{k-1}\rs=\bb{E} \ls \nu(\mu(q(k)))| \cl{F}_{k-1}\rs$, for any function $\nu$. Consequently,
\begin{align*}
\cl{R}_1(K)  \notag\\
&\hspace{-0.8cm}= \bb{E}_{\Phi} \ls \sum_{k=0}^{K-1} \mu_* - \nu(\mu_*) + \nu(\mu(q(k)))- \mu(q(k)) | \cl{F}_{k-1},\cl{E}\rs\notag\\
&\hspace{-0.8cm}= \bb{E}_{\Phi} \ls \sum_{k=0}^{K-1} \mu_* - \nu(\mu_*) + \nu(\mu(q(k)))- \mu(q(k)) | \cl{E}\rs,\label{eq:regret_bound1}
\end{align*}
using the law of total expectation. Furthermore, we choose
\begin{equation*}
\nu(\hat{\mu}_d(k-1)) =\min\lc 1,\max\lc 0,\hat{\mu}_d(k-1) + \epsilon(k,d)\rc\rc.
\end{equation*}
Therefore, under $\cl{E}$, we have $\mu_* < \nu(\mu_*)$, leading to
\begin{align*}
    \cl{R}_1(K)&\leq \bb{E}_{\Phi} \ls \sum_{k=0}^{K-1} \nu(\mu(q(k)))- \mu(q(k)) | \cl{E}\rs\\
    &\leq \bb{E}_{\Phi} \ls \sum_{k=0}^{K-1} \sum_{d=1}^D\mathbbm{1}(q(k)=p_d) 2\epsilon(k,d) \rs\\
    &= \bb{E}_{\Phi} \ls\sum_{d=1}^D \int_{0}^{\Gamma(K-1,d)} 2\sqrt{\frac{2 \log (1/\delta)}{z}}\dif z \rs\\
    &= \bb{E}_{\Phi} \ls\sum_{d=1}^D \sqrt{32\Gamma(K-1,d) \log (1/\delta)}\rs \\
    &= \sqrt{32KD \log (1/\delta)}.
\end{align*}
Hence, by choosing $\delta=1/\sqrt{K}$, we arrive at $\cl{R}_1(K)\leq \sqrt{64KD \log (K)}$.
Combining the above relation with \eqref{eq:regret_bound_new} gives
\begin{equation*}
    \cl{R}(K) \leq \sqrt{64KD \log (K)} + 4TD.
\end{equation*}
Hence, we arrive at the desired result.
\bibliography{references.bib}
\bibliographystyle{IEEEtran}
\vspace{-1cm}
\begin{IEEEbiography}{Gourab Ghatak} (Member, IEEE) received the Ph.D. degree from Telecom Paris Tech (University of Paris Saclay), France, during which he was with CEA Grenoble, France. Currently, he is an Assistant Professor with the Department of Electrical Engineering, Indian Institute of Technology Delhi
(IIT Delhi). His research interests include stochastic geometry, MACPHY cross-layer issues in beyond 5G systems, and machine learning for wireless communications.
\end{IEEEbiography}
\vspace{-1cm}
\begin{IEEEbiography}
    {Geethu Joseph} received the B. Tech. degree in electronics and communication engineering from the National Institute of Technology, Calicut, India, in 2011, and the M. E. degree in signal processing and the Ph.D. degree in electrical communication engineering (ECE) from the Indian Institute of Science (IISc), Bangalore, in 2014 and 2019, respectively. She was a postdoctoral fellow with the Department of Electrical Engineering and Computer Science at Syracuse University, NY, USA, from 2019 to 2021. She is currently a tenured assistant professor in the signal processing systems group at the Delft University of Technology, Delft, Netherlands. Dr. Joseph was awarded the 2022 IEEE SPS Best PhD dissertation award and the 2020 SPCOM Best Doctoral Dissertation award.  She is also a recipient of the Prof. I. S. N. Murthy Medal in 2014 for being the best M. E. (signal processing) student in the ECE dept., IISc, and the Seshagiri Kaikini Medal for the best Ph.D. thesis of the ECE dept. at IISc for the year 2019-'20. She is an associate editor of the IEEE Sensors Journal.  Her research interests include statistical signal processing, network control, and machine learning.
\end{IEEEbiography}
\vspace{-1cm}
\begin{IEEEbiography}
    {Chen Quan} received the B.S. degree in electronic engineering from the Nanjing University of Science and Technology, Nanjing, China, in 2016 and her M.S. degree in electronic engineering and PhD in electrical engineering from Syracuse University, Syracuse, NY, USA, in 2019 and 2023, respectively. She is currently a postdoctoral researcher in the signal processing systems group at the Delft University of Technology, Delft, Netherlands. Her research interests include tracking and optimization.
\end{IEEEbiography}

\end{document}